\documentclass[11pt]{article}
\usepackage[utf8]{inputenc}
\usepackage[margin=1in]{geometry}
\usepackage{amsthm}
\usepackage{amsmath}
\usepackage{amssymb}
\usepackage{url}
\usepackage{hyperref}

\newtheorem{lemma}{Lemma}
\newtheorem{example}{Example}

\newcommand{\floor}[1]{\lfloor #1 \rfloor}

\title{Nonlinear Arithmetic with SMTLIB Division is Undecidable}

\author{Dejan Jovanovi\'{c}\\
  Amazon Web Services\\
  \texttt{dejajov@amazon.com}
}

\begin{document}

\maketitle

\begin{abstract}
We show that the nonlinear real arithmetic theory (NRA) as defined in the SMTLIB standard is undecidable. The undecidability arises from the treatment of division by zero as an uninterpreted function, which allows encoding integer arithmetic problems into NRA formulas.

\medskip
\noindent\textbf{Note.} This is a short note that was written some time ago and is being archived here for reference. The reader interested in this topic is encouraged to consult the more comprehensive and authoritative treatment of division semantics in polynomial solvers by Brown \cite{brown2025semantics}, which subsumes the observations in this note.
\end{abstract}



%
%

\section{Undecidability of NRA}
\label{sect:undecidability}

As shown by Tarski \cite{tarski1998decision}, the theory of real closed fields is
decidable. Formulas of the theory of real closed fields are equalities and
inequalities over polynomials (polynomial constraints), commonly referred to in
the SMT community as nonlinear real arithmetic. The SMTLIB standard
\cite{barrett2010smt} formally defines the theory of nonlinear real arithmetic
(NRA) as the theory of the standard model of the reals $\mathbb{R}$. However, in
addition to polynomial constraints, the signature of NRA includes division. The
interpretation of division is as expected at all points except when the
divisor is $0$. When the divisor is $0$, division in NRA is uninterpreted in the sense that
any interpretation of $\frac{x}{0}$ is valid as long as division behaves as expected
when the divisor is nonzero. In other words, division is axiomatized as
\begin{align*}
  (\forall x, y \in \mathbb{R}) \;\; (y \neq 0) \Rightarrow (x = \frac{x}{y} \cdot y)\enspace.
\end{align*}

We show that the inclusion of division as described above leads to
undecidability of the SMTLIB theory NRA. This is perhaps unexpected for a theory whose decidable fragment (polynomial constraints) is a landmark result in mathematical logic.

\begin{lemma}
Real arithmetic with one unary uninterpreted function is undecidable.
\end{lemma}
\begin{proof}
Assume one uninterpreted function symbol $f$. We can axiomatize $f$ as the floor
function $\floor{\cdot}$ as follows
\begin{align}
(\forall x \in \mathbb{R}) \;\; & f(x) + 1 = f(x + 1)\enspace,\label{eq:floor1}\\
(\forall x \in \mathbb{R}) \;\; & (0 \leq x < 1) \Rightarrow (f(x) = 0)\enspace.\label{eq:floor2}
\end{align}
We denote the above axioms as $A$. By induction on $|\floor{x}|$, it is easy to
show that in the standard model of the reals, $f(x) = \floor{x}$. Let $F(x_1,
\ldots, x_n)$ be an arbitrary formula over integer variables $x_i$. Then $F$ is equisatisfiable with the following formula over real
variables
\begin{align*}
  A \wedge  \bigwedge_{i=1}^n (x_i = f(x_i)) \wedge F(x_1, \ldots, x_n)\enspace.
\end{align*}
If nonlinear real arithmetic with one uninterpreted function $f$ were decidable,
using the above transformation, we could decide
satisfiability of an arbitrary integer formula $F$. This is impossible since
integer arithmetic is undecidable (Hilbert's 10th problem
\cite{matiyasevich1993hilbert}).
\end{proof}

\begin{lemma}
Real arithmetic with uninterpreted division by zero is undecidable.
\end{lemma}
\begin{proof}
Using the previous lemma, we can use division by zero $\frac{x}{0}$ as the function
$f$, leading to undecidability.
\end{proof}

\begin{example}
The following example shows how to encode satisfiability of $a^3 + b^3 = c^3$
over integer variables $a$, $b$, and $c$, as a satisfiability problem in NRA.
\begin{align*}
(\forall x \in \mathbb{R}) \;\; & \frac{x}{0} + 1 = \frac{x + 1}{0}\enspace,\\
(\forall x \in \mathbb{R}) \;\; & (0 \leq x < 1) \Rightarrow (\frac{x}{0} = 0)\enspace,\\
& (\frac{a}{0} = a) \wedge (\frac{b}{0} = b) \wedge (\frac{c}{0} = c) \wedge (a^3 + b^3 = c^3)\enspace.
\end{align*}

\end{example}

\section{Discussion}
\label{sect:discussion}

\paragraph{Decision procedures.}

Although mostly focusing on linear arithmetic, there is a large body of work in
SMT on efficiently deciding quantified real arithmetic constraints. Notable
examples are approaches that apply complete quantifier elimination in a lazy
\cite{monniaux2010quantifier} or model-driven fashion (e.g.,
\cite{bjorner2015playing,reynolds2017solving}). These methods are appealing since,
in principle, they could be extended to decide quantified nonlinear constraints
by using, e.g., quantifier elimination techniques from CAD. Unfortunately, as
explained above, this is not possible in the presence of division by zero as an uninterpreted function.

\paragraph{Benchmarks.}

The NRA category of the SMTLIB library contains thousands of benchmarks. While
most use division only with constant divisors (which is decidable), benchmarks
with non-constant divisors fall into the undecidable fragment described in this
note.

\paragraph{Potential solutions.}

Two possible approaches to address this undecidability are:
\begin{enumerate}
\item Make division total by assigning a specific value to division by zero, similar to the treatment in the bitvector theory \cite{barrett2017qfbv}. The main challenge is choosing an appropriate value for $\frac{x}{0}$.
\item Move division with non-constant divisors to a separate theory (UFNRA) and keep NRA restricted to polynomial constraints plus division by non-zero constants. The downside is that this would require reclassifying many existing benchmarks.
\end{enumerate}

\label{sect:bib}
\bibliographystyle{plain}
\bibliography{nra}

\end{document}